\documentclass{article}

\PassOptionsToPackage{numbers, compress}{natbib}

\usepackage[final]{AdvML_Frontiers_2024}

\usepackage[utf8]{inputenc} 
\usepackage[T1]{fontenc}    
\usepackage{hyperref}       
\usepackage{url}            
\usepackage{booktabs}       
\usepackage{amsfonts}       
\usepackage{nicefrac}       
\usepackage{microtype}      
\usepackage{caption} 
\usepackage{enumitem}
\usepackage{algorithm2e}
\usepackage{graphicx}
\usepackage{subcaption}
\usepackage{makecell}
\usepackage{booktabs} 
\usepackage{multirow} 
\usepackage[table]{xcolor} 
\usepackage{amsmath}
\usepackage{amssymb}
\usepackage{mathtools}
\usepackage{amsthm}

\usepackage[capitalize,noabbrev]{cleveref}

\theoremstyle{plain}
\newtheorem{theorem}{Theorem}[section]

\newtheorem{lemma}[theorem]{Lemma}

\theoremstyle{definition}
\newtheorem{definition}[theorem]{Definition}

\theoremstyle{remark}
\newtheorem{remark}[theorem]{Remark}

\title{dSTAR: Straggler Tolerant and Byzantine Resilient Distributed SGD}

\author{%
  Jiahe Yan\\
  Department of Computer Science\\
  University of California, Los Angeles\\
  Los Angeles, CA 90024 \\
  \texttt{yjh020711@g.ucla.edu} \\
  \And
  Pratik Chaudhari\\
  Department of Electrical and System Engineering\\
  University of Pennsylvania\\
  Philadelphia, PA 19104 \\
  \texttt{pratikac@seas.upenn.edu} \\
  \And
  Leonard Kleinrock\\
  Department of Computer Science\\
  University of California, Los Angeles\\
  Los Angeles, CA 90024 \\
  \texttt{lk@cs.ucla.edu} \\
}

\begin{document}

\maketitle

\begin{abstract}
Distributed model training needs to be adapted to challenges such as the straggler effect and Byzantine attacks. When coordinating the training process with multiple computing nodes, ensuring timely and reliable gradient aggregation amidst network and system malfunctions is essential. To tackle these issues, we propose \textit{dSTAR}, a lightweight and efficient approach for distributed stochastic gradient descent (SGD) that enhances robustness and convergence. \textit{dSTAR} selectively aggregates gradients by collecting updates from the first \(k\) workers to respond, filtering them based on deviations calculated using an ensemble median. This method not only mitigates the impact of stragglers but also fortifies the model against Byzantine adversaries. We theoretically establish that \textit{dSTAR} is (\(\alpha, f\))-Byzantine resilient and achieves a linear convergence rate. Empirical evaluations across various scenarios demonstrate that \textit{dSTAR} consistently maintains high accuracy, outperforming other Byzantine-resilient methods that often suffer up to a 40-50\% accuracy drop under attack. Our results highlight \textit{dSTAR} as a robust solution for training models in distributed environments prone to both straggler delays and Byzantine faults.
\end{abstract}

\section{Introduction}
Distributed SGD has become a standard way of training large machine learning models due to its scalability and efficiency in processing vast amounts of data in parallel across multiple computing nodes. We consider the classical setting with a single parameter server and \(N\) workers \cite{abadi2016tensorflow}. Given \(X \in \mathbb{R}^{m \times d}\) representing \(m\) \(d\)-dimensional data, \(y \in \mathbb{Z}^{m}\) where each element of \(y\) is the discrete label of the respective row in \(X\), and a loss function \(F(\theta)\) for the dataset, where \(\theta\) are the model parameters, the parameter server wants to find \(\theta^*\) that minimizes the loss function \(F\). During each iteration, the parameter server sends model parameters \(\theta\) to all workers. Each worker contains a unique subset of \(X\) to parallelize gradient computation. The worker computes and returns the gradient of \(\theta\) on the local dataset to the server, which then aggregates the gradients to perform stochastic gradient descent.

While distributed SGD offers enhanced scalability and acceleration, it also introduces fault tolerance concerns in distributed systems. Workers in a distributed system can be Byzantine faulty. The identity of such Byzantine workers is also a priori unknown. Byzantine workers may produce wrong or even malicious results back to the parameter server due to various reasons, from system failure to malicious attacks \cite{lamport1982byzantine}. Averaging, which is the simplest way to aggregate gradients from workers, has been proven fragile to even one worker being Byzantine \cite{Blanchard2017MLAdversaries}. To confer Byzantine resilience in distributed SGD, many Gradient Aggregation Rules (GARs) have been proposed to allow learning to occur under a (maximum) number of \(f\) Byzantine workers under synchronous and asynchronous settings. The maximum \(f\) that can be tolerated is called the \textit{breakdown point}, with the \textit{optimal breakdown point} being \(N > 2f\) \cite{Blanchard2017MLAdversaries}. That is, as long as the majority of workers are honest, model training can proceed. However, these GARs come with their own challenges. In synchronous SGD, the parameter server needs to wait for slow or unresponsive nodes known as stragglers \cite{Dean2013TailScale}. In asynchronous SGD, the server will update the model parameter as soon as any worker returns a gradient to avoid stragglers \cite{Dean2012LargeScale}. However, this leads to a smaller batch size per aggregation, effectively introducing noise to the model. Additionally, the server may also receive ``stale gradients'' computed from outdated \(\theta\), potentially causing the model to converge more slowly or even diverge. 

To address the dual challenges of Byzantine resilience and straggler tolerance, we present \textit{dSTAR}, a new Byzantine-resilient distributed SGD that selectively waits for \(k\) gradients from the fastest workers, selected using a filter that calculates deviations of worker gradients from an ensemble median (where \(1\leq k \leq N\) and is adaptive). In fault-free settings, the fastest-\(k\) SGD (or formally, synchronous SGD with backup workers) has been shown to achieve optimal performance as synchronous SGD while mitigating the straggler effect \cite{Chen2016RevisitingDistributed}. In the fastest-\(k\) SGD, the parameter server only waits for the fastest \(k\) workers per iteration before making a gradient descent update. Other gradients will simply be dropped. If we assume the response time of each worker is \(i.i.d\), it can be shown that the fastest-\(k\) SGD is equivalent to the single-node batch SGD since the server updates based on a uniformly random set of gradients. However, the fastest-\(k\) SGD is vulnerable to Byzantine attacks. By definition, Byzantine workers can return gradients anytime they want, whereas the response time of a non-Byzantine worker can be unbounded. Hence, Byzantine workers can always be in the fastest \(k\) and compromise training. In this paper, we introduce a new fastest-\(k\) variant that can be robust under Byzantine attack as long as the majority of nodes are honest. We show that \textit{dSTAR} consistently produces optimal models under different Byzantine attacks, model architecture, and datasets while other GARs can experience performance drops of 40-50\(\%\). Furthermore, since \(k\) is adjustable, \textit{dSTAR} offers a configurable spectrum from fully asynchronous to fully synchronous operation. This flexibility allows for tailoring the system dynamics based on specific requirements and constraints of the deployment environment.

\section{Related work}
Formally, a GAR is robust to Byzantine attacks if it satisfies \((\alpha, f)\)-Byzantine resilience \cite{Blanchard2017MLAdversaries}:

\begin{definition}[\((\alpha, f)\)-Byzantine Resilience]
Let \(\alpha \in [0, \frac{\pi}{2}]\), \(f \in [0, n]\). Let \(V_1, V_2,\ldots, V_n\) be any independent identically distributed random vectors in \(\mathbb{R}^d\) such that \(V_i \sim G\), with \(\mathbb{E}[G] = \nabla F\). Let \(B_1, B_2,\ldots, B_f\) be any random vectors \(\in \mathbb{R}^d\), possibly dependent on the \(V_i\)’s. An aggregation algorithm \(A\) is said to be \((\alpha, f)\)-Byzantine resilient if, for any \(1 \leq j_1 < \ldots < j_f \leq n\), the vector \(A = A(V_1,\ldots,\underbrace{B_1}_{j_1},\ldots,\underbrace{B_f}_{j_f},\ldots,V_n)\)
satisfies:
1) \(\langle \mathbb{E}[A], \nabla F \rangle \geq (1 - \sin(\alpha)) \lVert \nabla F \rVert^2 > 0\), and
2) for any \( r \in \{2, 3, 4\} \), \( \mathbb{E} \| A \| ^r \) is bounded above by a linear combination of terms \( \mathbb{E} \| G \| ^{r_1}, \ldots, \mathbb{E} \| G \| ^{r_{n-1}} \) with \( r_1 + \ldots + r_{n-1} = r \).
\end{definition}


Existing GARs ensure \((\alpha, f)\)-Byzantine Resilience by employing robust statistics to identify candidate gradients to aggregate. Most GARs focus on the fully synchronous setting where all gradients will be collected before applying the aggregation rule. Examples of synchronous GARs are as follows: a). AKSEL averages a subset of gradients based on their squared distances to the coordinate-wise median \cite{Boussetta2021Aksel}, b). KRUM chooses the gradient with the smallest sum of Euclidean distances with neighbors \cite{Blanchard2017MLAdversaries}, c). CGE averages a subset of gradients with the smallest norms \cite{Gupta2020ByzantineSGD}, d). TrMean discards extreme values and aggregates the top \((N - b)\) gradients nearest to the median where \(b\) is a hyperparameter \cite{Yin2018ByzantineRobust}. A few algorithms such as KARDAM and Zeno++ focus on the asynchronous setting, where the model can be updated as soon as any gradient is returned. KARDAM uses a sliding window based on gradient aggregation history and empirical Lipschitzness of gradients to filter for good gradients \cite{Damaskinos2018AsyncByzantineML}. Zeno++ chooses candidate gradients that lead to a greater descent of the loss value based on a validation set on the parameter server \cite{Xie2020Zeno}. Nevertheless, synchronous GARs suffer from stragglers and asynchronous GARs may produce suboptimal models. KARDAM can only support up to one-third of Byzantine workers. Zeno++ also requires manually configuring a gradient threshold, which can be can be time-consuming to optimize. Zeno++ further has a model error bound that is influenced by the presence of asynchronous noise, which can be substantial if stale gradients are utilized more than sparingly. Moreover, asynchronous GARs suffer from ``stale gradients'' computed from outdated model parameters.

\section{Contributions}
Traditional synchronous GARs mandate the collection of all workers' gradients for each iteration to ensure convergence because they depend on statistical measures within each iteration. To achieve optimal convergence without waiting for all gradients, \textit{dSTAR} focuses on statistics gained from the training history via a validation set approach similar to Zeno++. The parameter server keeps a unique subset of \(X\) as the validation set locally and computes its validation gradient to compare against incoming gradients. Unlike traditional approaches, \textit{dSTAR} determines a filtering threshold dynamically based on the historical ensemble median. \textit{dSTAR} further achieves optimal time complexity and breakdown point as shown in Table \ref{tab:method_comparison}. The key contributions of our work include: 1). Proposed a new SGD that addresses the straggler effect by waiting for only the \(k\) fastest gradients with a dynamically configured filtering threshold while being robust against Byzantine attacks; 2). Showed empirically that the SGD can consistently produce an optimal model; 3). Showed theoretically that the SGD has a linear convergence rate and is Byzantine-resilient.

\begin{table}[ht]
  \caption{Comparison of different gradient aggregation rules}
  \label{tab:method_comparison}
  \centering
  \begin{tabular}{lcc}
    \toprule
    \textbf{Method}       & \textbf{Time Complexity}     & \textbf{Breakdown Point} \\
    \midrule
    Average               & \(O(Nd)\)                      & \(f = 0\) \\
    AKSEL                 & \(O(Nd)\)                      & \(n > 2f\) \\
    TrMean                & \(O(Nd)\)                      & \(n > 2f\) \\
    KRUM                  & \(O(N^2d)\)                    & \(n > 2f + 1\) \\
    CGE                   & \(O(N (d + \log N))\)          & \(n > 2f\) \\
    \textbf{\textit{dSTAR}} & \(\mathbf{O(Nd)}\)            & \(\mathbf{n > 2f}\) \\
    \bottomrule
  \end{tabular}
\end{table}

\section{Assumptions}
\begin{enumerate}[leftmargin=*,label=\textbf{A\arabic*}]
\setlength\itemsep{-0.2em}
\item \textbf{(Unbiased gradients with bounded variance)} The proposed gradient \(g_i\) from the set of honest workers \(S_h\) are d-dimensional vectors and unbiased estimates of the true gradient and have bounded variance: 
\[ 
\forall i \in S_h, g_i \sim G, E[G] = \nabla F,
E[G_j - \nabla F_j]^2 = \sigma_j^2, E \|G - \nabla F\|^2 = E \sum_{j = 1}^{d} [G_j - \nabla F_j]^2 = d\sigma^2
\]

\item \textbf{(Lipschitz gradients)} The loss function \( F \) is Lipschitz continuous with \(L > 0 \): 
\[ \forall \theta_1, \theta_2, \left\| \nabla F(\theta_1) - \nabla F(\theta_2) \right\| \leq L \left\| \theta_1 - \theta_2 \right\| \]

\item \textbf{(Bounded gradients)} The gradients \(g_i\) from honest workers and \(g_v\) from validation set are all upper bounded by \(V\), the validation set gradient is also lower bounded by \(V'\) \cite{Xie2020Zeno}: 
\[\|g_i\|^2 \leq V,\  V' \leq \|g_v\|^2 \leq V,\ 0 < V' \leq V\]

\end{enumerate}

\section{Algorithm}
We present our new algorithm with its theoretical analysis. Algorithm \ref{alg:byzantine-resilient-aggregation} in the Appendix describes the full pseudo training loop code. \textit{dSTAR} aggregation involves evaluating each received gradient against two key metrics calculated from the validation gradient derived from the parameter server’s validation set: the dot product and the squared Euclidean distance. Unlike Zeno++, which requires manually configuring a threshold, \textit{dSTAR} compares both values against the values calculated using the historical median. The median in a system with optimal breakdown point is robust to Byzantine attack \cite{Xie2018GeneralizedByzantineTolerantSGD}. During the first iteration, we default to aggregate the median of all gradients (i.e. a fully synchronous iteration using MEDIAN GAR) since history is unknown. This procedure serves as a warm-up phase for the filtering of subsequent iterations. 

During each subsequent iteration \(t\), given an incoming gradient \(g^t_i\) and the local validation set gradient \(g^t_v\), the server computes normalized Euclidean distance \(s^t_i = \frac{\|g^t_i - g^t_v\|^2}{\|g^t_v\|}\) and dot product \(d^t_i = \left\langle \frac{g^t_i}{\|g^t_v\|}, \frac{g^t_v}{\|g^t_v\|} \right\rangle\). If \(s^t_i\) is less than or equal to the normalized Euclidean distance calculated using the historical median gradient and validation set gradient and \(d^t_i\) is greater than or equal to the normalized dot product calculated using the historical median gradient and validation set gradient, \(g^t_i\) is added to an accepted list. The collection phase stops once \(k\) gradients are accumulated or all workers have responded. Since gradients can vary significantly in magnitude across iterations, we included normalization in the calculation for Euclidean distance and the dot product to maintain a consistent scale relative to the validation gradient when evaluating the incoming gradients. The accepted gradients are then averaged to calculate the aggregated gradient \(g^t_{\text{agg}} = \frac{1}{k} \sum_{j=1}^{k} g^t_{\text{accepted}_j}\), and the model parameters are updated accordingly: \(\theta^{t+1} = \theta^t - \eta g^t_{\text{agg}}\). In experiments, we show that by simply using the first iteration median gradient and validation gradient as this historical threshold, \textit{dSTAR} already reaches top performance. In theory, extending the warmup period to more rounds may improve performance further. 

\subsection{Time complexity}
Calculating Euclidean distance and dot product are both \(O(Nd)\). For the first iteration, finding the median using quick select is also \(O(N)\) \cite{Blum1973TimeBoundsForSelection}. For all subsequent iterations, the algorithm simply retrieves the recorded median values and evaluates each incoming gradient against these metrics. Hence, the total time complexity for this algorithm is \(O (Nd) \). In practice, the effective time complexity is often lower than this theoretical upper bound as \(k < N\). Notably, \textit{dSTAR} has a much lower time complexity than methods like KRUM \((O(N^2d))\), which requires pairwise comparisons among gradients, and CGE \((O(N(d + \log N)))\), which requires sorting \(N\) gradients per iteration.

\subsection{Byzantine resilience analysis}
We show that \textit{dSTAR} is \((\alpha-f)\)-Byzantine resilient. First, it is important to point out the robustness of the median. For a sequence of higher-dimensional vectors with the optimal breakdown point, the coordinate-wise median will always lie within the range defined by the minimum and maximum values of the honest coordinates for that dimension \cite{Xie2018GeneralizedByzantineTolerantSGD}. Based on this, we illustrate that the aggregated gradient of each iteration satisfies the following two lemmas:

\begin{lemma}
(Proof in the appendix) Under assumptions A1 to A3, if \(g_*^t\) denotes the aggregated gradient for iteration \(t\), it satisifies:

\begin{equation}
\begin{split}
\langle \mathbb{E}[g^t_{*}], \nabla F \rangle &\geq \left( \| \nabla F \| - \sqrt{\frac{2(n-f)}{k}d\sigma^2} (\frac{V}{V^{'}})^{\frac{1}{4}} \right) \| \nabla F \|
\end{split}
\end{equation}
\end{lemma}

\begin{lemma}
(Proof in the appendix) Under assumptions A1 to A3, if \(g_*^t\) denotes the aggregated gradient for iteration \(t\), it is upper bounded by a linear combinations of \(\mathbb{E} \| G \|^{r1}, ...,\ \mathbb{E} \| G \|^{r_{n-1}}\)
\end{lemma}

Given the two lemmas, \textit{dSTAR} is \((\alpha-f)\)-Byzantine resilient under the optimal breakdown point:

\begin{theorem}
Let \(g^t_1, \ldots, g^t_n, g^t_v\) be i.i.d. \(d\)-dimensional gradients at iteration \(t\) such that \(g^t_i \sim G\), with \(\mathbb{E}[G] = \nabla F\) and \(\mathbb{E}\|G - \nabla F\|^2 = d\sigma^2\). \(f\) of \(\{g^t_1, \ldots, g^t_n\}\) are replaced by arbitrary values. The \(\text{dSTAR}\) function selects and aggregates \(g^t_1, \ldots, g^t_k\) where \(k \leq n\). If \(n > 2f\) and \(\sqrt{\frac{2(n-f)}{k}d\sigma^2} \left(\frac{V}{V^{'}}\right)^{\frac{1}{4}} < \|\nabla F\|\), then the \(\text{dSTAR}\) function is \((\alpha, f)\)-Byzantine resilient where \(0 \leq \alpha < \frac{\pi}{2}\) is defined by:
\begin{equation}
\begin{split}
\sin \alpha = \frac{\sqrt{\frac{2(n-f)}{k}d\sigma^2} \left(\frac{V}{V^{'}}\right)^{\frac{1}{4}}}{\|\nabla F\|}
\end{split}
\end{equation}
\end{theorem}

\begin{remark}
The condition on the norm of the gradient is standard in Byzantine resilience analysis \cite{Blanchard2017MLAdversaries}. It can be satisfied at least to some extent by computing gradients using mini-batches on workers. Averaging gradients over a mini-batch divides \( \sigma \) by the squared root of the mini-batch size \cite{Bottou1998OnlineLearning}. 
\end{remark}

\subsection{Convergence analysis}
\begin{theorem}
(Proof in appendix) Assume \(F(\theta)\) is L smooth, and there exists a global minimum \(\theta^*\) where \(F(\theta^*) \leq F(\theta)\ \forall_\theta\), then after training for \(T\) iterations, \textit{dSTAR} has expected error bound: \( \mathbb{E} \left[ F(\theta^*) - F(\theta^0) \right] \leq \sum_{t = 0}^{T} -\eta \frac{V}{V^{'}} \|\nabla F(\theta^{1})\|^2 + \mathbb{O} (V + d\sigma^2) \) where \(\nabla F(\theta^{t'})\) represents the honest gradient at certain iteration \(t'\).
\end{theorem}

\section{Experiments}
In this section, we detail the empirical evaluation of \textit{dSTAR}. We evaluated the algorithm and other synchronous GARs on two standard image classification benchmarks: Fashion-MNIST and CIFAR10, with LeNet-5 and ResNet18 architectures respectively. We assessed the resilience of each algorithm by subjecting them to two state-of-the-art Byzantine attacks: 
\begin{itemize}
    \item \textbf{``Little" \cite{Baruch2019ALittleIsEnough}:} The attack disrupts the median gradient computation by introducing spurious gradients that cluster around the mean. Specifically, given \(N\) workers in which \(f\) workers are Byzantine, the attack: 1). computes the number of required workers for a majority \(s = \left\lfloor \frac{N}{2} + 1 \right\rfloor - f\); 2). calculates the maximum \(z\)-value, \(z_{\max}\), from the standard normal distribution such that the cumulative probability \(\phi(z) < \frac{N - s}{N}\); 3). generates a malicious gradient \(g_{\text{mal}} = \mu + z_{\max} \cdot \sigma\), using the mean \(\mu\) and standard deviation \(\sigma\) of non-Byzantine gradients.
    \item \textbf{``Empire" \cite{Xie2020FallOfEmpires}:} The attack employs inner product manipulation to break Byzantine-tolerant GARs. The attack uses the fact that, for gradient descent algorithms to guarantee the descent of the loss, the inner product between the true gradient and the aggregated gradient must be non-negative. Hence, malicious gradients can be generated to make the aggregated gradient point in the opposite direction as the true gradient \( (g_{\text{mal}} = -s \mu) \) where \(\mu\) is the honest gradient mean and \(s\) is a configurable scaling factor.
\end{itemize}
We simulate a distributed environment with 25 workers and a Byzantine ratio of \(35\%\). Each worker contains a unique subset of the dataset, comprising random samples across all classes. Network delays are modeled using an exponential distribution with rate \(\beta\). Honest workers have \(\beta = 0.2\) and Byzantine workers have \(\beta = 0.001\). The value of \(\beta\) makes no difference for synchronous GARs because they need to wait for all nodes, but for \textit{dSTAR} it makes faulty workers significantly more likely to be in the fastest \(k\), thereby exposing the vulnerability of vanilla fastest-\(k\) algorithm. For \textit{dSTAR}, the initial \(k\) is set as 8, and the time to aggregate gradients in each iteration will be the time to accept \(k\) gradients or the maximum response time from all nodes if our filter cannot accept \(k\) gradients, in which case \textit{dSTAR} waits for all nodes to return but only aggregate the accepted ones. For all experiments, we used the Adam optimizer with an initial learning rate of 0.001. The preprocessing steps for Fashion-MNIST included converting images into tensors and normalizing them. For CIFAR10, images are padded on all sides with 4 pixels, randomly cropped into 32 \(\times\) 32 pixels, randomly flipped horizontally, and converted to tensors and normalized. Additionally, for CIFAR10, we implemented a cosine annealing scheduler to adjust the learning rate, with a minimum rate set at 0.0001. We also utilized Mixup for data augmentation with a parameter \(\alpha\) of 0.4. These preprocessing are added to make the fault-free baseline comparable to SOTA for accurate comparisons.

\section{Results}
In three of the four experiments, \textit{dSTAR} achieved top accuracy (see Table \ref{tab:empire_little_accuracy} and \ref{tab:empire_little_none_accuracy}). Furthermore, \textit{dSTAR} maintains a consistent performance across different Byzantine attacks, whereas other synchronous GARs may have up to 40-50\(\%\) drop between the two attacks. This uniformity in performance under various adversarial conditions underscores the robustness and generalized ability of \textit{dSTAR}. The performance of \textit{dSTAR} is particularly notable under the "Empire" attack scenarios, where it is the only algorithm that converges. The full training curves can be found in the Appendix Figures \ref{fig:fashion_empire_attack} to \ref{fig:cifar_little_attack}. 

Additionally, the goal of designing a fastest-\(k\) Byzantine resilient SGD is to mitigate the straggler effect. It has been shown in Table \ref{tab:time_comparison} that the selective waiting strategy for \(k\) fastest gradients significantly reduces the time required for gradient aggregation per iteration. The reduced wait times can contribute to higher throughput and efficiency, making \textit{dSTAR} particularly suited for time-sensitive applications.

The only setting where our proposed algorithm didn't achieve the best accuracy was on CIFAR10 under the "Little" attack, although the performance is still significantly better than TRMEAN and KRUM and is only 2\(\%\) lower than CGE. This can be explained by a tradeoff between accuracy and speed, as the accuracy will almost surely improve by waiting for more workers at the cost of a longer waiting time per iteration. Additionally, we default to MEDIAN for the initial iteration, which can be susceptible to the attack since "Little" was designed specifically for MEDIAN GAR. Choosing a different synchronous GAR for the initial iteration or having a longer warm-up phase may also improve performance.

\begin{table}[h]
  \caption{Fashion-MNIST accuracies of methods under Little and Empire attacks, including the fault-free baseline.}
  \label{tab:empire_little_accuracy}
  \centering
  \begin{tabular}{lccc}
    \toprule
    \textbf{Method} & \textbf{Little (\%)} & \textbf{Empire (\%)} & \textbf{Fault-Free (\%)} \\
    \midrule
    dSTAR              & \textbf{88.78} & \textbf{88.87} & 88.86 \\
    Trmean             & 16.55 & 32.48 & 89.44 \\
    Krum               & 88.19 & 40.84 & 88.22 \\
    CGE                & 88.30 & 82.44 & 89.47 \\
    Aksel              & 88.51 & 75.08 & 88.67 \\
    Average            & -    & -    & \textbf{89.65} \\
    \bottomrule
  \end{tabular}
\end{table}

\begin{table}[h]
  \caption{CIFAR10 accuracies of methods under Empire and Little attacks. The dashed columns indicate that the algorithm failed to converge. The fault-free baseline has an accuracy of 94.33\%.}
  \label{tab:empire_little_none_accuracy}
  \centering
  \begin{tabular}{lccc}
    \toprule
    \textbf{Method} & \textbf{Empire (\%)} & \textbf{Little (\%)} & \textbf{Fault-Free (\%)} \\
    \midrule
    dSTAR              & 91.11 & \textbf{91.60} & 91.23 \\
    Trmean             & 20.50 & 11.85 & 93.72 \\
    Krum               & 76.32 & 10.00 & 80.38 \\
    CGE                & \textbf{93.45} & 41.32 & 94.19 \\
    Aksel              & 92.44 & 46.62 & 93.64 \\
    Average            & -     & -     & \textbf{94.33} \\
    \bottomrule
  \end{tabular}
\end{table}

\begin{table}[h]
  \caption{Average time between iterations for synchronous GARs and \textit{dSTAR}}
  \label{tab:time_comparison}
  \centering
  \begin{tabular}{lc}
    \toprule
    \textbf{GAR} & \makecell{\textbf{Average Time}\\\textbf{Between Iterations (s)}} \\
    \midrule
    Synchronous GAR & 7.62 \\
    dSTAR           & \textbf{3.79} \\
    \bottomrule
  \end{tabular}
\end{table}

\section{Discussion and conclusion}
We introduced \textit{dSTAR}, a novel Byzantine resilient distributed SGD algorithm that effectively balances the dual challenges of mitigating straggler effects and defending against adversarial Byzantine attacks in synchronous settings. The experimental results demonstrated that \textit{dSTAR} is robust to various adversarial settings, whereas other synchronous GARs can have performance degradation when facing different Byzantine attacks. The ability of \textit{dSTAR} to deliver such results highlights its potential as a reliable solution for securing distributed SGD processes against an array of threats while ensuring minimal disruption to operational efficiency. 

Future work could involve scaling the experiments to more complex models and datasets to provide a more comprehensive understanding of the algorithm's performance and potential adjustments. Extending our experiments will help ascertain the generalizability of our findings across various domains and applications. Moreover, the integration of \textit{dSTAR} with emerging machine learning paradigms, such as federated learning, represents a promising research direction as well.

\bibliographystyle{unsrt} 
\bibliography{references} 






\appendix

\section{Appendix}

\subsection{Byzantine resilience analysis}
First, it is important to point out the robustness of the median. Formally, the median value given the optimal breakdown point is always bounded by two honest values and is Byzantine resilient. We restate Lemma 4 from \cite{Xie2018GeneralizedByzantineTolerantSGD} without proof:

\begin{lemma}
For a sequence composed of \( f \) Byzantine values and \( n - f \) honest values \( x_1, x_2, \ldots, x_{n-f} \), if \( f \leq \left\lceil \frac{n}{2} \right\rceil - 1 \) (the honest values dominate the sequence), then the median value \( m \) of this sequence satisfies \( m \in [x_{\min}, x_{\max}] \).
\end{lemma}

For a sequence of higher-dimensional vectors, the coordinate-wise median maintains the same robustness properties \cite{Xie2018GeneralizedByzantineTolerantSGD}. Specifically, the median for each coordinate will always lie within the range defined by the minimum and maximum values of the honest coordinates for that dimension. Following this lemma, we illustrate that \textit{dSTAR} is Byzantine resilient. For the first iteration, we default to MEDIAN aggregator which is already Byzantine resilient. For any subsequent iteration \(t\), we accept a gradient if its normalized Euclidean distance to the validation gradient of iteration \(t\) is not greater than the normalized Euclidean distance of the first iteration coordinate-wise median to the first iteration validation gradient. If we denote the first iteration coordinate-wise median as \(g_m\), the first iteration validation gradient as \(g_{v}^1\), the \(t\)-th iteration validation gradient \(g_v^t\), and an arbitrary gradient received from worker \(i\) during iteration \(t\) as \(g_i^t\), then we accept \(g_i^t\) if the following two inequalities hold:

\begin{equation}
\begin{split}
\frac{\|g_i^t - g_v^t\|^2}{\|g_v^t\|} \leq \frac{\|g_m - g_v^1\|^2}{\|g_v^1\|}
\end{split}
\end{equation}

\begin{equation}
\begin{split}
\langle \frac{g_m}{\|g_v^1\|}, \frac{g_v^1}{\|g_v^1\|} \rangle \leq \langle \frac{g_i^t}{\|g_v^t\|}, \frac{g^t_v}{\|g^t_v\|} \rangle
\end{split}
\end{equation}

From (3), we have:

\begin{gather}
\|g_i^t - g_v^t\|^2 \leq \frac{\|g_v^t\|}{\|g_v^1\|} \|g_m - g_v^1\|^2
\end{gather}

Assume all gradients are d-dimensional and come from the same distribution \(G\) where \(\mathbb{E} \left[ G_i - \nabla F_i \right]^2 = \sigma_i^2 \) and \(\mathbb{E} \| G - \nabla F \|^2 = \mathbb{E} \sum_{i = 1}^{d} \left[ G_i - g_i \right]^2 = d\sigma^2 \) , we have:

\begin{equation}
\begin{split}
\mathbb{E} \|g_m - g_v^1\|^2 &= \mathbb{E} [\sum_{j = 1}^d ((g_m)_j - (g^1_v)_j)^2] \\
&= \sum_{j = 1}^d \mathbb{E} [(g_m)_j - (g^1_v)_j]^2
\end{split}
\end{equation}

where \((g_m)_j\) represents the j-th dimension of the vector. Since \(g_m\) is the coordinate-wise median over first iteration gradients, we have \((g_m)_j \in \left[ \min_{\text{correct } i} (g^1_i)_j, \max_{\text{correct } i} (g^1_i)_j \right]\). We thus have:

\begin{equation}
\begin{split}
\mathbb{E} [(g_m)_j - (g_v^1)_j]^2 &\leq \mathbb{E} \left[ \max_{\text{correct } i} \left( (g^1_i)_j - (g_v^1)_j \right)^2 \right] \\
&\leq \mathbb{E} [\sum_{\text{correct } i} ((g^1_i)_j - (g_v^1)_j)^2] \\
&= \sum_{\text{correct } i} \mathbb{E} [((g^1_i)_j - (g_v^1)_j)^2] \\
&= (n - f) \mathbb{E} [((g^1_i)_j - (g_v^1)_j)^2] \\
&= (n - f) 2\sigma_j^2
\end{split}
\end{equation}

Thus, we can plug this back to (6) and obtain:

\begin{equation}
\begin{split}
\mathbb{E} \|g_m - g_v^1\|^2 &= \sum_{j = 1}^d \mathbb{E} [(g_m)_j - (g^1_v)_j]^2 \\
&\leq \sum_{j = 1}^d 2(n - f) \sigma^2_j \\
&\leq 2(n - f) d\sigma^2
\end{split}
\end{equation}

With Assumption A3, this gives us an upper bound for the expectation of (5):

\begin{equation}
\begin{split}
\mathbb{E} \|g^t_i - g^t_v\|^2 &\leq \mathbb{E} \frac{\|g_v^t\|}{\|g_v^1\|} \|g_m - g_v^1\|^2 \\
&\leq 2(n - f) \frac{\sqrt{V}}{\sqrt{V^{'}}}  d\sigma^2
\end{split}
\end{equation}

Now, we begin to prove the Byzantine Resilience of our algorithm. 

\begin{theorem}
Let \(g^t_1, \ldots, g^t_n, g^t_v\) be i.i.d. \(d\)-dimensional gradients at iteration \(t\) such that \(g^t_i \sim G\), with \(\mathbb{E}[G] = \nabla F\) and \(\mathbb{E}\|G - \nabla F\|^2 = d\sigma^2\). \(f\) of \(\{g^t_1, \ldots, g^t_n\}\) are replaced by arbitrary values. The \(\text{dSTAR}\) function selects and aggregates \(g^t_1, \ldots, g^t_k\) where \(k \leq n\). If \(n > 2f\) and \(\sqrt{\frac{2(n-f)}{k}d\sigma^2} \left(\frac{V}{V^{'}}\right)^{\frac{1}{4}} < \|\nabla F\|\), then the \(\text{dSTAR}\) function is \((\alpha, f)\)-Byzantine resilient where \(0 \leq \alpha < \frac{\pi}{2}\) is defined by:
\begin{equation}
\begin{split}
\sin \alpha = \frac{\sqrt{\frac{2(n-f)}{k}d\sigma^2} \left(\frac{V}{V^{'}}\right)^{\frac{1}{4}}}{\|\nabla F\|}
\end{split}
\end{equation}
\end{theorem}

\begin{proof}
We first focus on the condition (i) of Byzantine Resilience. Suppose we denote the final aggregated gradient during iteration \(t\) as \(g^t_{*}\), we want to determine an upper bound on \(\|\mathbb{E}[g^t_{*}] - \nabla F \|^2 \). If Assumption 1 holds, we have:

\begin{equation}
\begin{split}
\|\mathbb{E}[g^t_{*}] - \nabla F \|^2 &\leq \|\mathbb{E}\left( \ g^t_{*} - g^t_v \right)\|^2 \\
&\leq \mathbb{E} \|\ g^t_{*} - g^t_v \|^2 \\
&= \mathbb{E} \| \frac{1}{k} \sum_{j = 1}^{k} (g^t_j - g^t_v) \|^2 \\
& \leq \frac{1}{k^2} \sum_{j = 1}^{k} \mathbb{E} \|g^t_j - g^t_v\|^2 \\
&\leq \frac{2(n - f)}{k} \frac{\sqrt{V}}{\sqrt{V^{'}}} d\sigma^2
\end{split}
\end{equation}

If \( \sqrt{\frac{2(n-f)}{k}d\sigma^2} (\frac{V}{V^{'}})^{\frac{1}{4}} \leq \| \nabla F \| \), \( \mathbb{E}[g^t_{*}]\) belongs to a ball centered at \( \nabla F \) with radius \(\sqrt{\frac{2(n-f)}{k}d\sigma^2} (\frac{V}{V^{'}})^{\frac{1}{4}} \). This implies:

\begin{equation}
\begin{split}
\langle \mathbb{E}[g^t_{*}], \nabla F \rangle &\geq \left( \| \nabla F \| - \sqrt{\frac{2(n-f)}{k}d\sigma^2} (\frac{V}{V^{'}})^{\frac{1}{4}} \right) \| \nabla F \| \\
&= (1 - sin \alpha) \| \nabla F \|^2
\end{split}
\end{equation}

So condition (i) of Byzantine Resilience holds when \( \sqrt{\frac{2(n-f)}{k}d\sigma^2} (\frac{V}{V^{'}})^{\frac{1}{4}} \leq \| \nabla F \| \). Now we focus on condition (ii). For an accepted gradient \(g_j^t\) at iteration \(t\) with validation gradient \(g_v^t\), there exists a constant \(C\) such that:

\begin{equation}
\begin{split}
\|g_j^t\| &\leq \|g_j^t - g_v^t\| + \|g_v^t\| \\
&\leq (\frac{V}{V^{'}})^{\frac{1}{4}} \|g_m - g_v^1\| + \|g_v^t\|
\end{split}
\end{equation}

\begin{equation}
\begin{split}
\|g_m - g_v^1\| &= \sqrt{\sum_{j = 1}^d [(g_m)_j - (g_v^1)_j]^2} \\
&\leq \sqrt{\sum_{j = 1}^d \max_{\text{correct } i} [(g_i^1)_j - (g_v^1)_j]^2} \\
&\leq \sqrt{\sum_{j = 1}^d \sum_{\text{correct } i} [(g_i^1)_j - (g_v^1)_j]^2} \\
&\leq \sqrt{\sum_{\text{correct } i}  \|g_i^1 - g_v^1\|^2} \\
&\leq C \sum_{\text{correct } i} \|g_i^1 - g_v^1\| \\
&\leq C \sum_{\text{correct } i} \|g_i^1\| + \|g_v^1\|
\end{split}
\end{equation}

Putting this back to (13), we have:

\begin{equation}
\begin{split}
\|g_k^t\| &\leq \|g_v^t\| + C \sum_{\text{correct } i} \|g_i^1\| + \|g_v^1\|
\end{split}
\end{equation}

Since all terms on the right side are from correct gradients, we can conclude that the norm of each accepted gradient can be bounded by the norm of honest gradients. By triangle inequality, \(\|g_{*}^t\| = \|\frac{1}{k} \sum_{k} g_k^t\| \leq \frac{1}{k} \sum_{k} \|g_k^t\|\). So \(\mathbb{E} \|g_{*}^t\|^r \) is upper bounded by linear combinations of \(\mathbb{E} \| G \|^{r1}, ...,\ \mathbb{E} \| G \|^{r_{n-1}}\). Because both conditions are met, we can conclude that dSTAR is Byzantine resilient.
\end{proof}

\subsection{Convergence analysis}
For iteration \(t\), we denote the k gradients \textit{dSTAR} collects as \(g^t = \{g^t_1, g^t_2, ..., g^t_k\}\) and the validation gradient is \(g_v^t\). From our assumptions and Byzantine Resilience proof, we know \(V' \leq \|g_v^t\|^2 \leq V\) and \(\|g^t\|^2 \leq CV \) for some constant C. Assume \(F(\theta)\) captures the loss of \(\theta\) and is L smooth, and there exists a global minimum \(\theta^*\) where \(F(\theta^*) \leq F(\theta)\ \forall_\theta\), we want to find the error bound for the expected difference \(\mathbb{E} [F(\theta^T) - F(\theta^*)] \) after training our model for \(T\) iterations, which can be derived using a similar approach as \cite{Xie2020Zeno}. \\

From smoothness, we have:
\begin{equation}
\begin{split}
F(\theta^t) \leq & F(\theta^{t - 1}) + \langle \nabla F(\theta^{t - 1}), \theta^t - \theta^{t - 1} \rangle + \frac{L}{2} \| \theta^t - \theta^{t - 1} \|^2 \\
\end{split}
\end{equation}

For gradient descent update, \(\theta^t = \theta^{t - 1} - \eta g^t \):

\begin{equation}
\begin{split}
F(\theta^t) &\leq F(\theta^{t - 1}) + \langle \nabla F(\theta^{t - 1}), - \eta g^t \rangle + \frac{L}{2} \| - \eta g^t \|^2 \\
&\leq F(\theta^{t - 1}) + \langle \nabla F(\theta^{t - 1}), - \eta g^t \rangle + \frac{L\eta^2}{2} \| g^t \|^2 \\
&\leq F(\theta^{t - 1}) + \langle \nabla F(\theta^{t - 1}), - \eta g^t \rangle + \frac{L\eta^2}{2} CV,
\end{split}
\end{equation}

Now we focus on the dot product term:

\begin{equation}
\begin{split}
\langle \nabla F(\theta^{t - 1}), - \eta g^t \rangle &= \langle \nabla F(\theta^{t - 1}) - g^{t}_{v} + g^{t}_{v}, - \eta g^t \rangle \\
&= \langle \nabla F(\theta^{t - 1}) - g^{t}_{v}, - \eta g^t \rangle + \langle g^{t}_{v}, - \eta g^t \rangle \\
&\leq \eta \|\nabla F(\theta^{t - 1}) - g^t_v\|\ \|g^t\| + \langle g^t_v, - \eta g^t \rangle \\
&\leq \frac{\eta}{2} \|\nabla F(\theta^{t - 1}) - g^t_v\|^2 + \frac{\eta}{2} CV + \langle g^t_v, - \eta g^t \rangle
\end{split}
\end{equation}

Using triangle inequality, we know:

\begin{equation}
\begin{split}
\|\nabla F(\theta^{t - 1}) - g^t_v\|^2 &\leq 2\|\nabla F(\theta^{t - 1}) - \nabla F(\theta^{t})\|^2 + \\
&\quad 2\|\nabla F(\theta^{t}) - g^t_v\|^2
\end{split}
\end{equation}

\begin{equation}
\begin{split}
\langle \nabla F(\theta^{t - 1}), - \eta g^t \rangle &\leq \eta \|\nabla F(\theta^{t - 1}) - \nabla F(\theta^t)\|^2 + \\
&\quad \eta \|\nabla F(\theta^t) - g^t_v\|^2 + \frac{\eta}{2} CV + \\
&\quad \langle g^t_v, - \eta g^t \rangle
\end{split}
\end{equation}

From Assumption 1, \(\mathbb{E} \left\| \nabla F(\theta^t) - g^t_v \right\|^2 \leq d\sigma^2 \):
\begin{equation}
\begin{split}
\langle \nabla F(\theta^{t - 1}), - \eta g^t \rangle &\leq \eta \|\nabla F(\theta^{t - 1}) - \nabla F(\theta^t)\|^2 + \\
&\quad \eta d\sigma^2 + \frac{\eta}{2} CV + \langle g^t_v, - \eta g^t \rangle
\end{split}
\end{equation}

Using smoothness, we know:
\begin{equation}
\begin{split}
\|\nabla F(\theta^{t - 1}) - \nabla F(\theta^t)\|^2 &\leq L^2 \|\theta^{t - 1} - \theta^t\|^2 \\
&\leq L^2 \|\eta g^t\|^2 \\
&\leq L^2\eta^2 CV
\end{split}
\end{equation}

Now \(\langle \nabla F(\theta^{t - 1}), - \eta g^t \rangle\) is upper bounded by:
\begin{equation}
\begin{split}
\langle \nabla F(\theta^{t - 1}), - \eta g^t \rangle &\leq L^2\eta^3CV + \eta d\sigma^2 + \frac{\eta}{2} CV + \langle g^t_v, - \eta g^t \rangle
\end{split}
\end{equation}

From (4), we know the dot product \( \langle g^t_k, g^t_v \rangle \) for each accepted gradient \(g^t_k\) is guaranteed to be lower bounded by \( \frac{V}{V^{'}} \langle g_{m}, g^{1}_{v} \rangle \). We have:

\begin{equation}
\begin{split}
\langle g^t_v, g^t \rangle &= \frac{1}{k} \sum_{j = 1}^{k} \langle g^t_v, g^t_j \rangle \\
& \geq \frac{1}{k} \sum_{j = 1}^{k} \frac{V}{V^{'}} \langle g_{m}, g^{1}_{v} \rangle \\
&= \frac{V}{V^{'}} \langle g_{m}, g^{1}_{v} \rangle
\end{split}
\end{equation}

\begin{equation}
\begin{split}
\langle g^t_v, -\eta g^t \rangle &\leq -\eta \frac{V}{V^{'}} \langle g_{m}, g^{1}_{v} \rangle
\end{split}
\end{equation}

\begin{equation}
\begin{split}
\langle \nabla F(\theta^{t - 1}), - \eta g^t \rangle &\leq 
- \eta \left[ \frac{V}{V^{'}} \langle g_{m}, g^{1}_{v} \rangle - \left( L^2 \eta^2 + \frac{1}{2} \right) CV - d\sigma^2 \right]
\end{split}
\end{equation}

Plugging this back to (15), we get:
\begin{equation}
\begin{split}
F(\theta^t) &\leq F(\theta^{t - 1}) -\eta \frac{V}{V^{'}} \langle g_{m}, g^{1}_{v} \rangle + \\
&\quad \left( \frac{L\eta^2}{2} + L^2 \eta^3 + \frac{1}{2} \eta \right) CV + \eta d\sigma^2
\end{split}
\end{equation}

\begin{equation}
\begin{split}
F(\theta^t) - F(\theta^{t - 1}) \leq -\eta \frac{V}{V^{'}} \langle g_{m}, g^{1}_{v} \rangle + \mathbb{O} (V + d\sigma^2)
\end{split}
\end{equation}

Since \((g_m)_j \in \left[ \min_{\text{correct } i} (g^1_i)_j, \max_{\text{correct } i} (g^1_i)_j \right]\) and gradients are \(i.i.d\), we have:

\begin{equation}
\begin{split}
\langle g_m, g_v^1 \rangle &= \sum_{j = 1}^{d} (g_m)_j (g_v^1)_j \\
&\geq \sum_{j = 1}^{d} (g^1_i)_j (g_v^1)_j\ \text{for some correct } i \text{ on each dimension} \\
\end{split}
\end{equation}

\begin{equation}
\begin{split}
\mathbb{E} \langle g_m, g_v^1 \rangle &\geq \sum_{j = 1}^{d} \mathbb{E}[(g^1_i)_j (g_v^1)_j]\ \text{for some correct } i \text{ on each dimension} \\
&= \sum_{j = 1}^{d} (\nabla F(\theta^1)_j)^2 \\
&= \|\nabla F(\theta^1)\|^2
\end{split}
\end{equation}

By telescoping and taking the expectation of (28) and using the lower bound in (30), after \(T\) iterations we have:

\begin{equation}
\begin{split}
\mathbb{E} \left[ F(\theta^*) - F(\theta^0) \right] \leq \sum_{t = 0}^{T} -\eta \frac{V}{V^{'}} \|\nabla F(\theta^{1})\|^2 + \mathbb{O} (V + d\sigma^2)
\end{split}
\end{equation}

\newpage

\subsection{Full algorithm}

\begin{algorithm}[htp]
\caption{Byzantine-Resilient Gradient Aggregation}\label{alg:byzantine-resilient-aggregation}
\SetAlgoLined
\KwData{initial parameters \(\theta^0\), dataset \(\{X, y\}\), number of workers \(N\), initial k value \(k_0\), increase k threshold \(\tau_k\), batch size \(n_b\), Byzantine ratio \(f\), fixed step size \(\eta\), number of iterations \(T\), loss function \(F\)}
\KwResult{final parameters \(\theta^T\)}

\textbf{Initialization:}\\
\(k \leftarrow k_0\), \( S \leftarrow \) squared Euclidean distance for first iteration median \( g_i \), \( D \leftarrow \) dot product for first iteration median \( g_i \)\;
Create worker subsets and validation set \(X_i, V \neq X_i\) for each worker \(i \in [1, N]\)\;

\For{\(t \leftarrow 1\) \KwTo \(T\)}{
    \textbf{Step 1: Broadcast current estimate}\\
    The server broadcasts the current estimate \(\theta^t\) to all workers\;
    
    Each worker \(i\): \\
    Receives \(\theta^t\) from the server\;
    \eIf{worker \(i\) is honest}{
        Draws random batch \(x \sim X_i\), computes \(g_i = \frac{1}{n_b} \sum_{x_i \in x} \nabla F(\theta^t; x_i)\)\;
    }{
        Computes spurious gradient \(g_i\) as per Byzantine strategy\;
    }
    Sends \(g_i\) back to the server\;
    
    \textbf{Step 2: \textit{dSTAR} Aggregation}\\
    The server waits for gradients to be returned\;
    \For{each returned gradient \(g_i\)}{
        Compute normalized Euclidean distance \(s_i\) and dot product \(d_i\) with respect to the validation set gradient \(g_V\)\;
        \(s_i \leftarrow \frac{\|g_i - g_v\|^2}{\|g_v\|}\)\;
        \(d_i \leftarrow \langle \frac{g_i}{\|g_v\|}, \frac{g_v}{\|g_v\|} \rangle\)\;

        \If{\(s_i \leq S\) \textbf{and} \(d_i \geq D\)}{
            Append \(g_i\) to the accepted list\;
        }
    
        \If{\(k\) gradients have been appended}{
            \textbf{break}\;
        }
    }
    \(g_{\text{agg}} = \frac{1}{k} \sum_{j=1}^{k} g_{\text{accepted}_j}\)\;
    Update model parameters \(\theta^{t+1} \leftarrow \theta^t - \eta g_{\text{agg}}\)\;
    \(t \leftarrow t + 1\)\;
}
\label{alg:byzantine_sgd}
\end{algorithm}
\newpage

\subsection{Training curves for experiments}
\begin{figure}[htbp]
  \centering
  \includegraphics[width=0.8\linewidth]{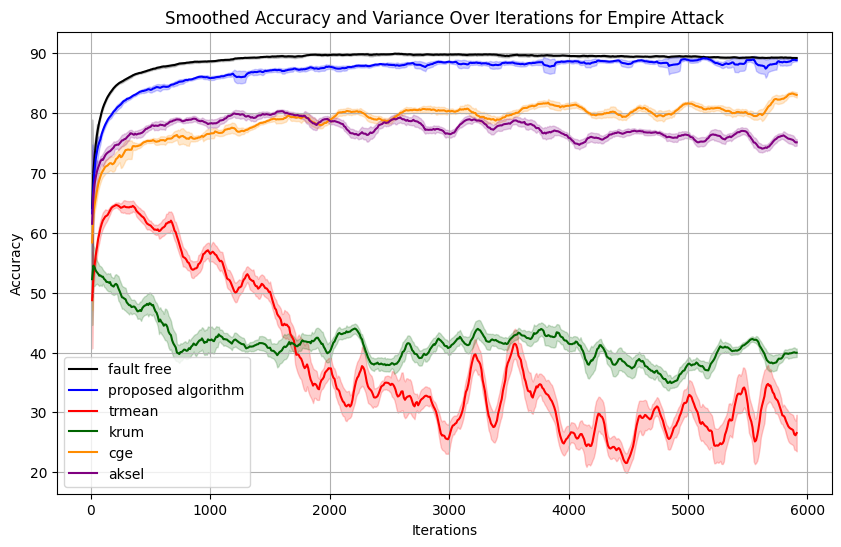}
  \caption{Fashion-MNIST with ``Empire" attack}
  \label{fig:fashion_empire_attack}
\end{figure}

\begin{figure}[htbp]
  \centering
  \includegraphics[width=0.8\linewidth]{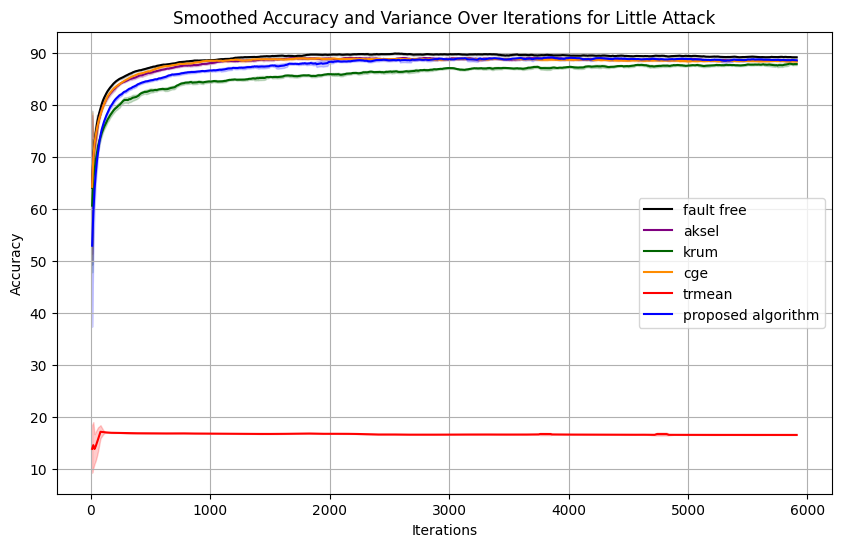}
  \caption{Fashion-MNIST with ``Little" attack}
  \label{fig:fashion_little_attack}
\end{figure}

\begin{figure}[htbp]
  \centering
  \includegraphics[width=0.8\linewidth]{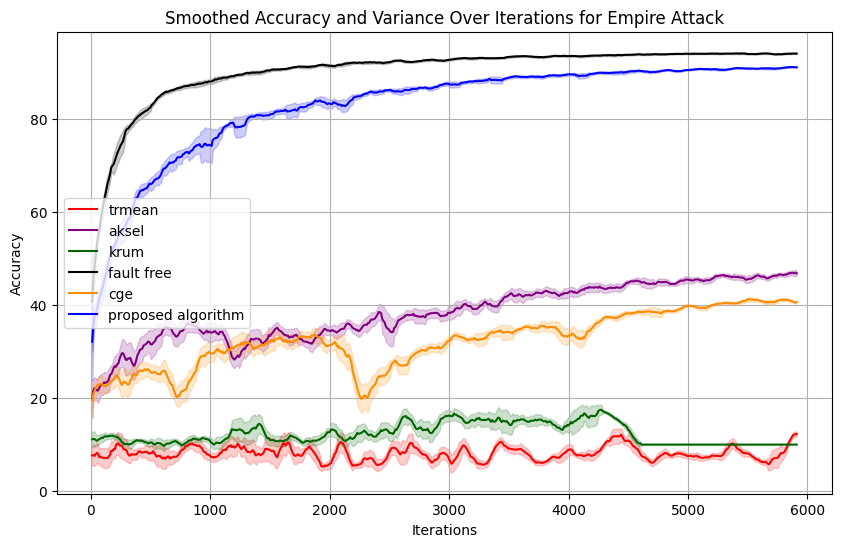}
  \caption{CIFAR10 with ``Empire" attack}
  \label{fig:cifar_empire_attack}
\end{figure}

\begin{figure}[htbp]
  \centering
  \includegraphics[width=0.8\linewidth]{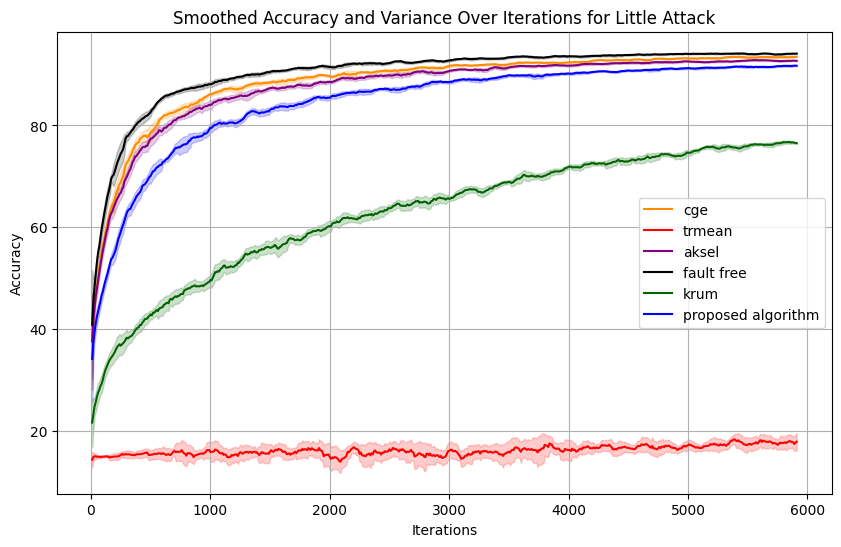}
  \caption{CIFAR10 with ``Little" attack}
  \label{fig:cifar_little_attack}
\end{figure}

\end{document}